\documentclass[10pt,letterpaper]{article}
\usepackage[top=0.85in,left=2.75in,footskip=0.75in]{geometry}

\usepackage{changepage}

\usepackage[utf8]{inputenc}

\usepackage{textcomp,marvosym}

\usepackage{fixltx2e}

\usepackage{amsmath,amssymb}
\usepackage{amsthm}

\usepackage{cite}

\usepackage{nameref,hyperref}

\usepackage{microtype}
\DisableLigatures[f]{encoding = *, family = * }

\usepackage{rotating}

\raggedright
\usepackage[aboveskip=1pt,labelfont=bf,labelsep=period,justification=raggedright,singlelinecheck=off]{caption}

\makeatletter
\renewcommand{\@biblabel}[1]{\quad#1.}
\makeatother

\date{}

\usepackage{lastpage,fancyhdr,graphicx}
\usepackage{epstopdf}

\usepackage[crop=pdfcrop]{pstool}
\usepackage{dsfont}
\usepackage{xargs}
\usepackage{ifthen}
\usepackage{adjustbox}
\usepackage{multirow}
\usepackage[textwidth=5cm,disable]{todonotes}
\usepackage{pgfplots}

\newcommand{\internaltodo}[1]{}
\newcommand{\todosr}[1]{\todo[author=stefan,backgroundcolor=blue!20!white]{#1}}

\newcommand{\N}{\mathds{N}}
\newcommand{\R}{\mathds{R}}

\newcommand{\UF}{\mathcal{U}}
\newcommand{\F}{\mathfrak{F}}

\newcommand{\eps}{\varepsilon}
\newcommand{\ND}{\mathcal{N}}
\newcommand{\set}[1]{\left\{#1\right\}}
\newcommand{\abs}[1]{\left|#1\right|}

\newcommand{\To}{\mathop{\rightarrow}}
\newcommandx{\E}[2][2=\empty]{\ifthenelse{\equal{#2}{\empty}}{\mathrm{E}}{\mathrm{E}_{#2}}\!\left(#1\right)}
\newcommand{\Var}[1]{\mathrm{Var}\left(#1\right)}
\newcommand{\Prob}[1]{\mathrm{Pr}\!\left\{#1\right\}}

\renewcommand{\vec}{\bm}
\newcommand{\supp}{\text{supp}}

\renewcommand{\vec}{\mathbf}

\theoremstyle{plain}
\newtheorem{lem}{Lemma}
\newtheorem{thm}{Theorem}
\theoremstyle{definition}

\newtheorem{defn}{Definition}
\newtheorem{exa}{Example}
\newtheorem{rem}{Remark}
\renewenvironmentx{proof}[1][1=\empty]{\ifthenelse{\equal{#1}{\empty}}{\emph{Proof.}~}{\emph{Proof (#1).~}}}{\hfill$\square$\\[0.1\baselineskip]}

\fancyheadoffset[L]{2.25in}
\fancyfootoffset[L]{2.25in}
\reversemarginpar
\newcommand{\rot}[1]{\begin{turn}{-90}#1\end{turn}}

\newcounter{daggerfootnote}
\newcommand*{\daggerfootnote}[1]{%
	\setcounter{daggerfootnote}{\value{footnote}}%
	\renewcommand*{\thefootnote}{\fnsymbol{footnote}}%
	\footnote[2]{#1}%
	\setcounter{footnote}{\value{daggerfootnote}}%
	\renewcommand*{\thefootnote}{\arabic{footnote}}%
}

\begin{document}

\vspace*{0.35in}

\begin{flushleft}
{\Large \textbf\newline{Decisions with Uncertain Consequences -- A Total
Ordering on Loss-Distributions}\daggerfootnote{This is a correction to the published version, available from \url{https://journals.plos.org/plosone/article?id=10.1371/journal.pone.0168583}. The correction has been submitted to the journal office in Feb. 2022, and is pending for publication, in which case this version will become replaced.}}
\newline
\\
Stefan Rass\textsuperscript{1,3,*}, %
Sandra K\"{o}nig\textsuperscript{2}, %
Stefan Schauer\textsuperscript{2}, %
\\
\bigskip
\bf{1} Universit\"{a}t Klagenfurt, Department of Artificial Intelligence and Cybersecurity,
Klagenfurt, Austria
\\
\bf{2} Austrian Institute of Technology, Safety \& Security Department,
Klagenfurt, Austria
\\
\bf{3} LIT Secure and Correct Systems Lab, Johannes Kepler University Linz, Austria
\\
\bigskip

%
%





* stefan.rass@aau.at

\end{flushleft}
\section*{Abstract}
Decisions are often based on imprecise, uncertain or vague information.
Likewise, the consequences of an action are often equally unpredictable, thus
putting the decision maker into a twofold jeopardy. Assuming that the effects
of an action can be modeled by a random variable, then the decision problem
boils down to comparing different effects (random variables) by comparing
their distribution functions. Although the full space of probability
distributions cannot be ordered, a properly restricted subset of
distributions can be totally ordered in a practically meaningful way. We call
these \emph{loss-distributions}, since they provide a substitute for the
concept of loss-functions in decision theory. This article introduces the
theory behind the necessary restrictions and the hereby constructible total
ordering on random loss variables, which enables decisions under uncertainty
of consequences. Using data obtained from simulations, we demonstrate the
practical applicability of our approach.




\section{Introduction}\todo{the introduction has been rewritten almost
entirely, towards a better connection with the text and to be shorter and
more condensed.}In many practical situations, decision making is a matter of
urgent and important choices being based on vague, fuzzy and mostly empirical
information. While reasoning under uncertainty in the sense of making
decision with known consequences under uncertain preconditions is a
well-researched field (cf.
\cite{Shortliffe1990,Pearl1988,Buntine1995,Jensen2002,Halpern2003,Evans2004,Koski2009}
to name only a few), taking decisions with \emph{uncertain consequences} has
received substantially less attention. This work presents a decision
framework to take the best choice from a set of options, whose consequences
or benefit for the decision maker are available only in terms of a random
variable. More formally, we describe a method to choose the best among two
possible random variables $R_1, R_2$ by constructing a novel stochastic order
on a suitably restricted subset of probability distributions. Our ordering
will be total, so that the preference between two actions with random
consequences $R_1, R_2$ is always well-defined and a decision can be made. As
it has been shown in \cite{Szekli1995,Stoyan2002}, there exist several
applications where such a framework of decision making on abstract spaces of
random variables is needed.

To illustrate our method, we will use a couple of example data sets, the
majority of which comes from the risk management context. In risk management, decisions
typically have uncertain consequences that cannot be measured by a
conventional von Neumann-Morgenstern utility function. For example, a
security incident in a large company can either be made public, or kept
secret. The uncertainty in this case is either coming from the public
community's response, if the incident is made public (as analyzed by, e.g.,
\cite{Busby2016}), or the residual risk of information leakage (e.g., by
whistleblowing). The question here is: Which is the better choice,
given that the outcomes can be described by random variables? For such a
scenario, suitable methods to determine the consequence distributions using
simulations are available \cite{Busby2016}, but those methods don't support
the decision making process directly.

Typically, risk management is concerned with extreme events, since small
distortions may be covered by the natural resilience of the analyzed system
(e.g., by an organization's infrastructure or the enterprise itself, etc.). For this reason, decisions normally
depend on the distribution's tails. Indeed, heavy- and fat-tailed distributions
are common choices to model rare but severe incidents in general risk
management \cite{Embrechts2003,McNeil2005}. We build our construction with
this requirement of risk management in mind, but originate from the recognized
importance that the moments of a distribution play for decision making (cf.
\cite{Eichner&Wagener2011}). In section \ref{sec:the-decision-framework}, we
illustrate a simple use of the first moment in this regard that is
common in IT risk management, to motivate the need to include more
information in a decision. Interestingly, the ordering that we define here is
based on the full moment sequence (cf. Definition
\ref{def:preference}), but implies similar conditions as other stochastic
orders, only with an explicit focus on the probability mass located in the
distribution's tails (cf. Theorem \ref{thm:tail-bounds}). Further, we pick
some example data sets from risk management applications in Section
\ref{sec:empirical-examples}, and demonstrate how a decision can be made
based on empirical data.

The main contribution of this work is twofold: while any stochastic order
could be used for decision making on actions with random variables describing
their outcome, not all of them are equally suitable in a risk management
context. The ordering we present in this article is specifically designed to
fit into this area. Second, the technique of constructing the ordering is new
and perhaps of independent scientific interest having applications beyond our
context. In the theoretical parts, this work is a condensed version of
\cite{Rass2015a,Rass2015b}, whereas it extends this preliminary research by
practical examples and concrete algorithms to efficiently choose best actions
despite random consequences and with a sound practical meaning.

\section{Preliminaries and Notation}\todo{this section has been made less verbose (to tighten this bit)}

Sets, random variables and probability distribution functions are denoted as
upper-case letters like $X$ or $F$. Matrices and vectors are denoted as
bold-face upper- and lower-case letters, respectively. The symbols $\abs{X},
\abs{x}$ denote the cardinality of the finite set $X$ or the absolute value
of the scalar $x\in\R$. The $k$-fold cartesian product (with $k=\infty$
permitted) is $X^k$, and $X^\infty$ is the set of all infinite sequences
$(a_n)_{n\in\N}=(a_1, a_2, a_3, \ldots)$ over $X$. Calligraphic letters like
$\F$ denote families (sets) of sets or functions. The symbol $^*\R$ denotes
the space of hyperreal numbers, being a certain quotient space constructed as
$^*\R=\R^\infty\slash\mathcal{U}$, where $\mathcal{U}$ is a free ultrafilter.
We refer to \cite{Robinson1966,Rass2015c} for details, as $^*\R$ is only a
technical vehicle whose detailed structure is less important than the fact
that it is a totally ordered field. Our construction of a total ordering on
loss distributions will crucially hinge on an embedding of random variables
into $^*\R$, where a natural ordering and full fledged arithmetic are already
available without any further efforts.

The symbol $X\sim F_X$ means the random variable (RV) $X$ having distribution
$F_X$, where the subscript is omitted if things are clear from the context.
The density function of $F_X$ is denoted by its respective lower-case letter
$f_X$. We call an RV \emph{continuous}, if it takes values in $\R$, and
\emph{discrete}, if it takes values on a countably infinite set $X$. A
\emph{categorical} RV is one with only finitely many, say $n$, distinct
outcomes. In that case, the density function can be treated as a vector $\vec
f_X\in\R^n$. 


\section{The Decision Framework}\label{sec:the-decision-framework}

Our decision problems will concern choosing actions of minimal \emph{loss}.
Formally, if $A$ is a set of actions, from which we ought to choose the best
one, then a \emph{loss-function} is usually some mapping $L:A\to\R^+$, so
that an optimal choice from $A$ is one with minimal loss under $L$ (see
\cite{Robert2001} for a full-fledged treatment and theory in the context of
Bayesian decision theory). In IT risk management\internaltodo{@Sandra: hier
war der Zusatz zu IT Risikomanagement wichtig, da diese Formel sicher nicht
ueberall so zum Einsatz kommt...da muessen wir vorsichtig sein...} (being
used to illustrate
our methods later in Section \ref{sec:empirical-examples}), 
risk is often quantified by
\begin{equation}\label{eqn:it-risk-formula}
    \text{ risk } = \text{damage} \times \text{ likelihood},
\end{equation}
which roughly resembles the idea of understanding risk as the expectation of
damage.
In this quantitative approach, the damage is
captured by the aforementioned loss function $L$, whereas the likelihood is
obtained from the distribution of the random event causing the damage.

However, losses can not always be measured precisely.
For the introductory example, consider the two actions $a_1=$
``publish the incident'' and $a_2=$``keep the incident secret''. Either
choice has unpredictable consequences so we replace the deterministic loss-function by a random variable.
That is,  let $a_1,a_2\in A$ be two arbitrary actions, and write
$X:=L(a_1)$ and $Y:=L(a_2)$, respectively, for the \emph{random} losses implied by taking these actions.
The challenge now is to make a decision that minimizes the risk when losses are random.

Obviously, comparing $X$ and $Y$ in the way suggested by
\eqref{eqn:it-risk-formula} has some shortcomings, as it is easy to construct random variables
with equal mean but highly different variance (the same issue
	would also exist in game theory \cite{Gibbons1992}, where the utility of
	mixed strategies is exactly the expectation of outcomes but normally
	disregards further moments). For the example of two
Gaussian variables $X\sim\ND(5,1), Y\sim\ND(5,10)$, the expectations are
equal, but actions resulting in losses measured by $Y$ are undesirable
relative to $X$, since the fluctuation around the mean for $Y$ is
considerably larger than for actions with consequences described by $X$.
An apparent quick fix is to take the variance into account for a decision.
However, the previous issue is still not mitigated, since it is equally easy
(yet only slightly more involved) to construct two random variables with
equal first and second moment, but with different third moments (Example
\ref{exa:equal-first-two-moments} will give two such distributions
explicitly). Indeed, the third moment can be taken into account in the
straightforward way, which has been discussed in the literature on risk
attitudes; see \cite{Eichner&Wagener2011,Chiu2010,Wenner2002} for a few
starting references.
Towards a more sophisticated
approach, we will in the following use the whole object (the random variable) rather that
a few representative values thereof to make a decision.

\subsection{The Usual Stochastic Order $\leq_{st}$}
Choosing a best action among $\set{a_1,a_2}$, we ought to compare the random
variables $X,Y$ in some meaningful way. Without any further restrictions on
the support or distribution, we may take the \emph{usual stochastic order}
\cite{Shaked2006} $\leq_{st}$ for that purpose, which calls $X\leq_{st} Y$ if
and only if
\begin{equation}\label{eqn:usual-stochastic-order}
  \Pr(X>x)\leq \Pr(Y>x)\quad\text{ for all }x\in(-\infty,\infty).
\end{equation}
Condition \eqref{eqn:usual-stochastic-order} can be stated equivalently by
demanding $\E{\phi(X)}\leq \E{\phi(Y)}$ for all increasing
functions $\phi$ for which the expectations
exist (so-called test-functions). In the latter formulation, it is easy to see that, for example, the
$\leq_{st}$-ordering in particular entails $\E{X}\leq \E{Y}$, so that a
comparison based on \eqref{eqn:it-risk-formula} comes out the same under
$\leq_{st}$. Moreover, in restricting $X$ and $Y$ to take on only positive
values, as our above definition of $L:A\to\R^+$ implies, $X\leq_{st} Y$
implies that all moments are in pairwise $\leq$-order, since the respective
functions $\phi(x)=x^k$ delivering them are all increasing on $\R^+$. Under
this restriction, comparisons based on the second and third moment
\cite{Chiu2010} are also covered under $\leq_{st}$.

\subsection{Generalizing $\leq_{st}$: The $\preceq$-Ordering}
In cases where it is sufficient to lower risk under an acceptance threshold,
rather than truly minimizing them, we may indeed relax the
$\leq_{st}$-ordering in several ways: we can require
\eqref{eqn:usual-stochastic-order} only for large damages in $(x_0,\infty)$
for a threshold $x_0$ that may be different for various application domains,
or we may not use all increasing functions, but only a few selected ones (our
construction will use the latter and entail the former relaxation). Given
that moments are being used to analyze risks and are related to risk
attitudes \cite{Chiu2010}, let us take the functions $\phi(x)=x^k$ for
$k\in\N$, which are all increasing on $\R^+$. To assure the existence of all
moments $\E{\phi(X)}<\infty$ and the monotony of all members in our
restricted set of test-functions, we impose the following assumptions on a
general random variable $R$, which we hereafter use to quantitatively model
``risk'':

\begin{defn}\label{asm:finity-of-repair-costs}
Let $\F$ be the set of all random variables $R$, who satisfy the following
conditions:
\begin{itemize}
  \item $R$ has a known distribution $F$ with compact support (note that
      this implies that $R$ is upper-bounded).
  \item $R\geq 1$ (w.l.o.g., since as $R$ is bounded, we can shift it into
      the region $[1,\infty)$).
  \item The probability measure induced by $F$ is either discrete or
      continuous and has a density function $f$. For continuous random
      variables, the density function is assumed to be continuous, and
      piecewise polynomial over a finite partition of the
      support.\todosr{die stueckweise polnomielle Dichte auf einer
      endlichen Zerlegung war ja der wesentliche Reparaturansatz}
\end{itemize}
\end{defn}
Requirement 1 assures that all moments exist. Requirements 2 and 3 serve
technical reasons that will be made clear in Lemma
\ref{lem:ordering-invariance}. In brief, these two assure that the ordering
obtained will be total, and simplifies proofs by defining the order as equal
to the natural ordering of hyperreal numbers. The requirement of the density
to be piecewise polynomial over a finite number of segments is necessary to
avoid families of distributions with alternating moments, such as were
constructed by \cite{burgin_remarks_2021}. Such families include even some
benign members with monotone densities, and have an order that explicitly
depends on the ultrafilter, which we do not want (this will be made rigorous
in Theorem \ref{thm:ordering-invariance} below). Assuming that the density,
if it is continuous, has a finite piecewise polynomial definition, excludes
these pathological cases. This exclusion is mild, since it can be shown
\cite[Lemma 3.1]{rass_game_2021} that most distributions can be approximated
to arbitrary precision by distributions satisfying the requirements of
Definition \ref{asm:finity-of-repair-costs}. The permission to restrict our
attention to moments rather than the whole random variable is given by the
following well known fact:

\begin{lem}\label{lem:uniqueness}
Let two random variables $X,Y$ have their moment generating functions
$\mu_X(s), \mu_Y(s)$ exist within a neighborhood $U_\eps(0)$. Assume that
$m_X(k)$ $:=$ $\E{X^k}$ $=$ $\E{Y^k}=:m_Y(k)$ for all $k\in\N$. Then $X$ and
$Y$ have the same distribution.
\end{lem}
\begin{proof}[Sketch] The proof is a simple matter of combining
well-known facts about power-series and moment-generating functions (see
\cite{Rass2015b} for a description).
\end{proof}
In the following, let us write $m_X(k)$ to mean the $k$-th moment of a random
variable $X$. Our next lemma establishes a total relation (so far not an
ordering) between two random variables from $\F$, on which our ordering will
be based:

\begin{lem}\label{lem:ordering-invariance}
For any two probability distributions $F_1,F_2$ and associated random
variables $R_1\sim F_1, R_2\sim F_2$ according to Definition
\ref{asm:finity-of-repair-costs}, there is a $K\in\N$ so that either
$[\forall k\geq K: m_{R_1}(k)\leq m_{R_2}(k)]$ or $[\forall k\geq K:
m_{R_1}(k)\geq m_{R_2}(k)]$.
\end{lem}
The proof of lemma \ref{lem:ordering-invariance} is given in the appendix.
The important fact stated here is that between any two random variables
$R_1,R_2$, either a $\leq$ or a $\geq$ ordering holds \emph{asymptotically}
on the moment sequence. Hence, we can take Lemma
\ref{lem:ordering-invariance} to justify the following relaxation of the
usual stochastic order:

\begin{defn}[$\preceq$-Preference Relation over Probability Distributions]\label{def:preference}
Let $R_1, R_2\in \F$ be two random variables with distribution functions
$F_1,F_2$. We \emph{prefer} $R_1$ \emph{over} $R_2$, respectively the
distribution $F_1$ over $F_2$, written as
\begin{equation}\label{eqn:partial-ordering}
    R_1\preceq R_2\iff F_1\preceq F_2:\iff \exists K\in\N\text{ s.t. }\forall k\geq K: m_{R_1}(k)\leq m_{R_2}(k)
\end{equation}
\emph{Strict preference} is denoted and defined as
\[
    R_1\prec R_2 \iff F_1\prec F_2:\iff \exists K\in\N\text{ s.t. }\forall k\geq K: m_{R_1}(k)<m_{R_2}(k)
\]
\end{defn}
For this definition to be a meaningful ordering, we need to show that
$\preceq$ behaves like other orderings, say $\leq$ on the real numbers. We
get all useful properties almost for free, by establishing an isomorphy
between $\preceq$ and another well known ordering, namely the natural $\leq$
order on the hyperreal space $^*\R$:
\begin{thm}\label{thm:ordering-invariance}
Let $\F$ be according to definition \ref{asm:finity-of-repair-costs}. Assume
every element $X\in\F$ to be represented by hyperreal number $\vec
x=(\E{X^k})_{k\in\N}\in \R^\infty\slash\UF$, where $\UF$ is any free
ultrafilter. Let $X,Y\in\F$ be arbitrary. Then, $X \preceq Y$ if $\vec
x\leq\vec y$ in $^*\R$, irrespectively of $\UF$.
\end{thm}
\begin{proof}[cf. \cite{Rass2015a}]
Let $F_1, F_2$ be two probability distributions, and let $R_1\sim F_1,
R_2\sim F_2$. Lemma \ref{lem:ordering-invariance} assures the existence of
some $K\in\N$ so that $F_1\preceq F_2$ iff $m_{R_1}(k)\leq m_{R_2}(k)$
whenever $k\geq K$. Let $L$ be the set of indices where $m_{R_1}(k)\leq
m_{R_2}(k)$, then complement set $\N\setminus L$ is finite (it has at most
$K-1$ elements). Let $\UF$ be an arbitrary free ultrafilter. Since
$\N\setminus L$ is finite, it cannot be contained in $\UF$ as $\UF$ is free.
And since $\UF$ is an ultrafilter, it must contain the complement a set,
unless it contains the set itself. Hence, $L\in\UF$, which implies the claim.
\end{proof}
Theorem \ref{thm:ordering-invariance} has quite some useful implications:
first, the asserted independence of the ultrafilter $\UF$ spares us the need
to explicitly construct $\UF$ (note that the general question of whether or
	not non-isomorphic hyperreal fields would arise from different choices of
	ultrafilters is still unanswered by the time of writing this article). Second, the $\preceq$-ordering on $\F$ inherits all properties (e.g.,
transitivity) of the natural ordering $\leq$ on $^*\R$, which by the transfer
principle \cite{Robinson1966}, hold in the same way as for $\leq$ on $\R$.
More interestingly for further applications, topological properties of the
hyperreals can also be transferred to $\F$. This allows the definition of a
whole game theory on top of $\preceq$, as was started in \cite{Rass2015c}. It
must be noted, however, that the $\preceq$-ordering still behaves different
to $\leq$ on $\R$, since, for example, the equivalence-relation induced by
$\preceq$ does not entail an identity between distributions (since a finite
number of moments is allowed to mismatch in any case).

Interestingly, although not demanded in first place, the use of moments to
compare a distribution entails a similar fact as inequality
\eqref{eqn:usual-stochastic-order} upon which the usual stochastic order was
defined:
\begin{thm}\label{thm:tail-bounds}
Let $X,Y\in\F$ have the distributions $F_1, F_2$. If $X\preceq Y$, then there
exists a threshold $x_0\in\supp(F_1)\cup\supp(F_2)$ so that for every $x\geq
x_0$, we have $\Pr(X>x)\leq \Pr(Y>x)$.
\end{thm}
The proof of this appears in the appendix. Intuitively, Theorem
\ref{thm:tail-bounds} can be rephrased into saying that:
\begin{quote}
If $F_1\preceq F_2$, then ``extreme events'' are less likely to occur under
$F_1$ than under $F_2$.
\end{quote}
Summarizing the results obtained, we can say that the $\preceq$-ordering
somewhat resembles the initial definition of the usual stochastic order
$\leq_{st}$, up to the change of restricting the range from
$(-\infty,\infty)$ to a subset of $[1,\infty)$ and in allowing a finite
number of moments to behave arbitrarily. Although this allows for an explicit
disregard of the first few moments, the overall effect of choosing a
$\preceq$-minimal distribution is shifting all the probability mass towards
regions of lower damages, which is a consequence of Theorem
\ref{thm:tail-bounds}. As such, this result could by itself be taken as a
justification to define this ordering in first place. However, in the way
developed here, the construction roots in moments and their recognized
relation to risk attitudes \cite{Eichner&Wagener2011,Chiu2010,Wenner2002},
and in the end aligns itself to both, the intuition behind $\leq_{st}$ and
the focus of risk management on extreme events, without ever having stated
this as a requirement to begin with. Still, by converting Theorem
\ref{thm:tail-bounds} into a definition, we could technically drop the
assumption of losses being $\geq 1$. We leave this as an aisle for future
research. As a justification of the restrictions as stated, note that most
risk management in the IT domain is based on categorical terms (see
\cite{Bundestag2010,BundesamtBevoelkerungsschutz2013,NANS2011,SCCA2012,Mell2007}),
which naturally map into integer ranks $\geq 1$. Thus, our assumption seems
mild, at least for IT risk management applications (applications in other
contexts like insurance \cite{Hogg&Klugman1984} are not discussed here and
constitute a possible reason for dropping the lower bound in future work).

\subsection{Distributions with Unbounded Tails}\label{sec:unbounded-tails}
Theorem \ref{thm:tail-bounds} tells that distributions with thin tails would
 be preferred over those with fat tails. However, catastrophic events are usually modeled by distributions with fat, heavy or
long tails. The boundedness condition in definition \ref{asm:finity-of-repair-costs} rules out
many such distributions relevant to risk management (e.g., financial risk
management \cite{Baeuerle2006}).
Thus, our next step is extending the ordering by
relaxing some of the assumptions that characterize $\F$.

The $\preceq$-relation
cannot be extended to cover distributions with heavy tails, as those typically do not have finite moments
or moment generating functions. For example, L\'{e}vi's $\alpha$-stable
distributions \cite{Nolan2016} are not analytically expressible as densities
or distribution functions, so the expression $\E{\phi(X)}$ could be quite
difficult to work out for the usual stochastic order. Conversely, resorting
to moments, we can work with characteristic functions, which can be much more
feasible in practice.

Nevertheless, such distributions are important tools in risk management.
Things are, however, not drastically restricted, for at least two reasons:

\begin{enumerate}
  \item Compactness of the support is not necessary for all moments to
      exist, as the Gaussian distribution has moments of all orders and is
      supported on the entire real line (thus violating even two of the
      three conditions of assumption \ref{asm:finity-of-repair-costs}).
      Still, it is characterized entirely by its first two moments, and
      thus can easily be compared in terms of the $\preceq$-relation.
  \item Any distribution with infinite support can be approximated by a
      truncated distribution.~\todo{the explanations about truncated
      distributions have been shortened, as this is standard and a compact
      reminder may suffice to follow the text} Given a random variable $X$
      with distribution function $F$, then \emph{truncated distribution}
      $\hat F$ is the conditional likelihood $\hat F(x) = \Pr(X\leq x|a\leq
      X\leq b)$.



By construction, the truncated distribution has the compact support
$[a,b]$. More importantly, for a loss distribution with unbounded support
$[1,\infty)$ and given any $\eps>0$, it is easy to choose a compact
interval $[a,b]$ large enough inside $[1,\infty)$ so that $\abs{F(x)-\hat
F(x)}<\eps$ for all $x$.
Hence, restricting ourselves to distributions with compact support, i.e.,
adopting assumption \ref{asm:finity-of-repair-costs}, causes no more than a
numerical error that can be made as small as we wish.
\end{enumerate}

More interestingly, we could attempt to play the same trick as before, and
characterize a distribution with fat, heavy or long tails by a sequence of
approximations to it, arising from better and better accuracy $\eps\To 0$. In
that sense, we could hope to compare approximations rather than the true
density in an attempt to extend the preference and equivalence relations
$\preceq$ and $\equiv$ to distributions with fat, heavy or long tails.

Unfortunately, such hope is an illusion, as a distribution is not uniquely
characterized by a general sequence of approximations (i.e., we cannot
formulate an equivalent to lemma \ref{lem:uniqueness}), and the outcome of a
comparison of approximations is not invariant to how the approximations are
chosen (i.e., there is also no alike for lemma
\ref{lem:ordering-invariance}). To see the latter, take the quantile function
$F^{-1}(\alpha)$ for a distribution $F$, and consider the tail quantiles
$\overline{F}^{-1}(\alpha) = F^{-1}(1-\alpha)$. Pick any sequence
$(\alpha_n)_{n\To\infty}$ with $\alpha_n\To 0$. Since
$\lim_{x\To\infty}F(x)=1$, the tail quantile sequence behaves like
$\overline{F}^{-1}(\alpha_n)\To \infty$, where the limit is independent of
the particular sequence $(\alpha_n)_{n\To\infty}$, but only the speed of
divergence is different for distinct sequences.

Now, let two distributions $F_1, F_2$ with infinite support be given.  Fix
two sequences $\alpha_n$ and $\omega_n$, both vanishing as $n\To\infty$, and
set
\begin{equation}\label{eqn:support-sequences}
  a_n := \overline{F}_1^{-1}(\alpha_n) \leq  b_n := \overline{F}_2^{-1}(\omega_n).
\end{equation}
Let us approximate $F_1$ by a sequences of truncated distributions $\hat
f_{1,n}$ with supports $[1,a_n]$ and let the sequence $\hat f_{2,n}$
approximate $f_2$ on $[1,b_n]$. Since $a_n< b_n$ for all $n$, it is easily
verified that the sequence of moments of the distributions truncated to
$[1,a_n]$ and $[1,b_n]$ implies that the respective moment sequences diverge
so that $\hat f_{1,n}\preceq \hat f_{2,n}$ ultimately. However, by replacing
the ``$< $'' by a ``$>$'' in \eqref{eqn:support-sequences}, we can construct
approximations to $F_1, F_2$ whose truncated supports overlap one another in
the reverse way, so that the approximations would always satisfy $\hat
f_{1,n}\succeq \hat f_{2,n}$. It follows that the sequence of approximations
\emph{cannot} be used to unambiguously compare distributions with infinite
support, unless we impose some constraints on the tails of the distributions
and the approximations. The next lemma (see appendix \ref{apx:criterion} for
a proof) assumes this situation to simply not occur, which allows to give a
\emph{sufficient} condition to unambiguously extend strict preference in the
way we wish.

\begin{lem}\label{lem:approximation-comparison}
Let $F_1,F_2$ be two distributions supported on $[1,\infty)$ with continuous
densities $f_1,f_2$. Let $(a_n)_{n\in\N}$ be an arbitrary sequence with
$a_n\To\infty$ as $n\to\infty$, and let $\hat f_{i,n}$ for $i=1,2$ be the
truncated distribution $f_i$ supported on $[1,a_n]$.

If there is a constant $c<1$ and a value $x_0\in\R$ such that $f_1(x)<c\cdot
f_2(x)$ for all $x\geq x_0$, then there is a number $N$ such that all
approximations $\hat f_{1,n},\hat f_{2,n}$ satisfy $\hat f_{1,n}\prec \hat
f_{2,n}$ whenever $n\geq N$.
\end{lem}

By virtue of lemma \ref{lem:approximation-comparison}, we can extend the
strict preference relation to distributions that satisfy the hypothesis of
the lemma but need not have compact support anymore. Precisely, we would
strictly prefer one distribution over the other, if all truncated
approximations are ultimately preferable over one another.

\begin{defn}[Extended Preference Relation
$\prec$]\label{def:extended-strict-preference} Let $F_1, F_2$ be distribution
functions of nonnegative random variables that have infinite support and
continuous density functions $f_1, f_2$. We \emph{(strictly) prefer} $F_1$
\emph{over} $F_2$, denoted as $F_1\prec F_2$, if for every sequence
$a_n\to\infty$ there is an index $N$ so that the approximations $\hat
F_{i,n}$ for $i=1,2$ satisfy $\hat F_{1,n}\prec \hat F_{2,n}$ whenever $n\geq
N$.

The $\succ$-relation is defined alike, i.e., the ultimate preference of $F_2$
over $F_1$ on any sequence of approximations.
\end{defn}

Definition \ref{def:extended-strict-preference} is motivated by the above
arguments on comparability on common supports, and lemma
\ref{lem:approximation-comparison} provides us with a handy criterion to
decide the extended strict preference relation.

\begin{exa}\label{exa:evd-comparison} It is a matter of simple algebra to
verify that any two out of the three kinds of extreme value distributions
(Gumbel, Frechet, Weibull) satisfy the above condition, thus are strictly
preferable over one another, depending on their
particular parametrization. 
\end{exa}

Definition \ref{def:extended-strict-preference} can, however, not applied to
every pair of distributions, as the following example shows.
\begin{exa}
Take the ``Poisson-like'' distributions with parameter $\lambda>0$,
\[
    f_1(k) \propto \left\{
               \begin{array}{ll}
                 \frac{\lambda^{k/2}}{(k/2)!}e^{-\lambda}, & \hbox{when $k$ is even;} \\
                 0, & \hbox{otherwise.}
               \end{array}
             \right.,\quad
    f_2(k) \propto \left\{
               \begin{array}{ll}
                 0, & \hbox{when $k$ is even};\\
\frac{\lambda^{(k-1)/2}}{((k-1)/2)!}e^{-\lambda}, & \hbox{otherwise} \\
               \end{array}
            \right.
\]
It is easy to see that no constant $c<1$ can ever make $f_1<c\cdot f_2$ and
that all moments exist. However, neither distribution is preferable over the
other, since finite truncations to $[1,a_n]$ based on the sequence $a_n:=n$
will yield alternatingly preferable results.
\end{exa}

An occasionally simpler condition that implies the hypothesis of definition
\ref{def:extended-strict-preference} is
\begin{equation}\label{eqn:strict-preference-alternative-criterion}
  \lim_{x\To\infty}\frac{f_1(x)}{f_2(x)} = 0.
\end{equation}
The reason is simple: if the condition of definition
\ref{def:extended-strict-preference} were violated, then there is an infinite
sequence $(x_n)_{n\in\N}$ for which $f_1(x_n)\geq c\cdot f_2(x_n)$ for all
$c<1$. In that case, there is a subsequence $(x_{n_k})_{k\in\N}$ for which
$\lim_{k\To\infty} f_1(x_{n_k})/f_2(x_{n_k})\geq c$. Letting $c\To 1$, we can
construct a further subsequence of $(x_{n_k})_{k\in\N}$ to exhibit that
$\limsup_{n\To\infty}(f_1(x_n)/f_2(x_n))=1$, so that condition
\eqref{eqn:strict-preference-alternative-criterion} would be refuted. Observe
that \eqref{eqn:strict-preference-alternative-criterion} is similar to the
definition of a likelihood ratio order \cite{Shaked2006} in the sense that it
implies both, a likelihood ratio and $\prec$-ordering. Note that, however, a
likelihood ratio order does not necessarily imply a $\prec$-order, since the
former only demands $f(t)/g(t)$ to be increasing, but not a $<$-relation
among the densities.

\begin{rem}\label{rem:strict-preference-not-extensible}
It must be emphasized that the above line of arguments does not provide us
with a mean to extend the $\preceq$- or $\equiv$-relations accordingly. For
example, an attempt to define $\preceq$ and $\equiv$ as above is obviously
doomed to failure, as asking for two densities $f_1, f_2$ to satisfy
$f_1(x)\leq c_1\cdot f_2(x)$ ultimately (note the intentional relaxation of
$<$ towards $\leq$), and $f_2(x)\leq c_2\cdot f_1(x)$ ultimately for two
constants $c_1,c_2<1$ is nonsense.
\end{rem}

A straightforward extension of $\preceq$ can be derived from (based on) the
conclusion of lemma \ref{lem:approximation-comparison}:
\begin{defn}\label{def:extended-preference} Let $F_1,F_2$ be two
distributions supported on the entire nonnegative real half-line $\R^+$ with
continuous densities $f_1,f_2$. Let $(a_n)_{n\in\N}$ be a diverging sequence
towards $\infty$, and let $\hat F_{i,n}$ for $i=1,2$ denote the density $F_i$
truncated to have support $[1,a_n]$. We define $F_1\preceq F_2$ if and only
if for every sequence $(a_n)_{n\in\N}$ there is some index $N$ so that $\hat
F_{1,n}\preceq \hat F_{2,n}$ for every $n\geq N$.
\end{defn}

More compactly and informally spoken, definition
\ref{def:extended-preference} demands preference on all approximations with
finite support except for at most finitely many exceptions near the origin.

Obviously, preference among distributions with finite support implies the
extended preference relation to hold in exactly the same way (since the
sequence of approximations will ultimately become constant when $a_n$
overshoots the bound of the support), so definition
\ref{def:extended-preference} extends the $\preceq$-relation in this sense.

\subsection{Comparing Distributions of Mixed Type}
The representation of a distribution by the sequence of its moments is of the
same form, for discrete, categorical and continuous random variables. Hence,
working with sequence representations (hyperreal numbers) admits to compare
continuous to discrete and categorical variables, as long as there is a
meaningful common support. The framework itself, up to the results stated so
far, remains unchanged and is applied to the category's ranks instead. The
ranking is then made in ascending order of loss severity, i.e., the category
with lowest rank (index) should be the one with the smallest damage magnitude
(examples are found in IT risk management standards like ISO 27005 \cite{ISO2011}
or the more generic ISO 31000 \cite{ISO2009} as well as related standards).

A comparison of mixed types is, obviously, only meaningful if the respective
random variables live in the same (metric) space. For example, it would be
meaningless to compare ordinal to numeric data. Some applications in natural
risk management define categories as numeric ranges (such as
\cite{Bundestag2010,BundesamtBevoelkerungsschutz2013,NANS2011,SCCA2012}),
which \emph{could} make a comparison of categories and numbers meaningful
(but not necessarily so).

\section{Practicalities}\label{sec:practicalities} It must be noted that
Definition \ref{def:preference} demands only the existence of some index
after which the sequence of moment diverges, without giving any condition to
assure this. Likewise, Theorem \ref{thm:tail-bounds} is non-constructive in
asserting the existence of a region onto which the $\preceq$-smaller
distribution puts more mass than the other. Hence, practical matters of
deciding and interpreting the $\preceq$-ordering are necessary and discussed
in the following.

In general, if the two distributions are supported on the sets $[1,a]$ for
$F_1$ and $[1,b]$ for $F_2$ with $b>a$, then the mass that $F_2$ puts on the
set $(a,b]$ will cause the moments of $F_2$ to grow faster than those of
$F_1$. In that case, we can thus immediately conclude $F_1\preceq
F_2$, and we get $x_0=a$ in Theorem \ref{thm:tail-bounds}. Thus, the more
interesting situation arises when the supports are identical, which is
assumed throughout the following subsections. Observe that it is herein not necessary to look
	at overlaps at the lower end of the supports, since the mass assigned near
	the ``right end'' of the support is what determines the growth of the moment
	sequence; the proof of Lemma \ref{lem:ordering-invariance} in the appendix
	more rigorously shows this.

\subsection{Deciding $\preceq$ between Categorical
Variables}\label{sec:preference-between-categorical-distr} Let $F_1, F_2$ be
two distributions over a common support, i.e., a common finite set of
categories, hereafter denoted in \emph{descending} order as $c_1>c_2>\ldots
>c_n$. Let $\hat f_1=(p_1,\ldots,p_n), \hat f_2=(q_1,\ldots, q_n)$ be the
corresponding probability mass functions. For example, these can be
normalized histograms (empirical density functions) computed from the
available data to approximate the unknown distributions $F_1, F_2$ of the
random variables $X,Y$.

Letting the category $c_i$ correspond to its rank $n-i+1$ within the support,
it is easy to check that the expectation of $X\sim F_1, Y\sim F_2$ by
definition is a sequence whose growth is determined by whichever distribution
puts more mass on categories of high loss. Formally, if $p_1>q_1$, then
$\E{X^k}=\sum_{j=1}^n p_j c_j^k
> \sum_{j=1}^n q_j c_j^k=\E{Y^k}$, since the growth of either sum is
determined by the largest term (here being $c_1^k$). Upon the equality
$p_1=q_1$, we can retract the respective terms from both sums (as they are
equal), to see whether the second-largest term $c_2^k$ tips the scale, and so
on.

Overall, we end up observing that $\preceq$-comparing distributions is quite
simple, and a special case of another common ordering relation:
\begin{defn}[lexicographic ordering]
For two real-valued vectors $\vec x = (x_1,x_2,\ldots)$ and $\vec
y=(y_1,y_2,\ldots)$ of not necessarily the same length, we define $\vec
x<_{lex} \vec y$ if and only if there is an index $i_0$ so that
$x_{i_0}<y_{i_0}$ and $x_i=y_i$ whenever $i<i_0$.
\end{defn}

Our discussion from above is then the mere insight that the following is
true:
\begin{thm}\label{thm:lexicographic-order}
Let $F_1, F_2$ be two categorical random variables with a common ordered
support $\Omega=\set{c_1> c_2>\ldots > c_n}$, and let $\vec f_1, \vec f_2$ be
the respective (empirical) density functions. Then $F_1\preceq F_2\iff \vec
f_1<_{lex} \vec f_2$, where $\vec f_i = (f_i(c_1), f_i(c_2), \ldots,
f_i(c_n))\in\R^n$.
\end{thm}

For illustration, we will apply Theorem \ref{thm:lexicographic-order} to two
concrete example data sets \#1 and \#2 in section
\ref{sec:practical-examples}.

\subsection{Deciding $\preceq$ between Continuous Variables}
Let us assume that the two random variables $R_1\sim F_1, R_2\sim F_2$ have
smooth densities $f_1,f_2\in C^{\infty}([1,a])$ for some $a>1$. Under this
assumption, we can switch to yet another useful sequence representation:
\begin{equation}\label{eqn:sequence-representation}
    f \mapsto \vec f = \left((-1)^k \hat f^{(k)}(a)\right)_{k\in\N}.
\end{equation}
Given two distributions $f_1, f_2\in C^\infty$, e.g., constructed from a
Gaussian kernel (cf. remark \ref{rem:convolution} below), let the respective
representations according to \eqref{eqn:sequence-representation} be $\vec
f_1, \vec f_2$. Then, it turns out that the lexicographic ordering of $\vec
f_1, \vec f_2$ implies the same ordering w.r.t. $\preceq$, or formally:
\begin{lem}[\!\!{\cite{Rass2015b}}]\label{lem:derivative-comparisons}
Let $f,g\in C^\infty([1,a])$ for a real value $a>1$ be probability density
functions. If
\[
    ((-1)^k\cdot f^{(k)}(a))_{k\in\N}<_{lex} ((-1)^k\cdot g^{(k)}(a))_{k\in\N},
\]
then $f\preceq g$.
\end{lem}

Lemma \ref{lem:derivative-comparisons} will be demonstrated on our example
data set \#3, in connection with a kernel density estimate, in section
\ref{sec:practical-examples}. Practically, we can thus decide the
$\preceq$-relation by numerically computing derivatives of increasing order,
until the decision is made by the lexicographic ordering (which, for our
experiments, happened already at zeroth order in many cases).

\begin{rem}\label{rem:convolution}
The assumption on differentiability is indeed mild, as we can cast any
integrable density function into a $C^\infty$-function by convolution with a
Gaussian density $k_h$ with zero mean and variance $h$. Clearly, $f\ast
k_h\in C^\infty$ by the differentiation theorem of convolution. Moreover,
letting $h\to 0$, we even have $L^1$-convergence of $f\ast k_h\to f$, so that
the approximation can be made arbitrarily accurate by choosing the parameter
$h>0$ sufficiently small. Practically, when the distributions are constructed
from empirical data, the convolution corresponds to a kernel density
estimation (i.e., a standard nonparametric distribution model). Using a
Gaussian kernel then has the additional appeal of admitting a closed form of
the $k$-th derivatives $(f\ast k_h)^{(k)}$, involving Hermite-polynomials.
\end{rem}

\paragraph{Observation -- ``$\preceq~\approx~<_{lex}$'':} As an intermediate
r\'{e}sum\'{e}, the following can be said:

\begin{quote}
Under a ``proper'' representation of the distribution (histogram or
continuous kernel density estimate), the $\preceq$-order can be decided as
a humble lexicographic order.
\end{quote}

This greatly simplifies matters of practically working with
$\preceq$-preferences, and also fits into the intuitive understanding of risk
and its formal capture by theorem \ref{thm:tail-bounds}: \emph{whichever
distribution puts more mass on far-out regions is less favourable under
$\preceq$.}

\subsection{Comparing Deterministic to Random Effects}
In certain occasions, the consequence of an action may result in perfectly
foreseeable effects, such as fines or similar. Such deterministic outcomes
can be modeled as degenerate distributions (point- or
Dirac-masses). These are singular and thus outside $\F$ by
Definition \ref{asm:finity-of-repair-costs}. Note that the canonic embedding of the reals within
	the hyperreals represents a number $a\in\R$ by the constant sequence
	$(a,a,\ldots)$. Picking up this idea would be critically flawed in our
	setting, as any such constant sequence would be preferred over any
	probability distribution (whose moment sequence diverges and thus overshoots
	$a$ inevitably and ultimately).

However, it is easy to work out
the moment sequence of the constant $X=a$ as $\E{X^k}=\E{a^k}=a^k$ for all
$k\in\N$. In this form, the $\preceq$-relation between the number $a$ and the
continuous random variable $Y$ supported on $\Omega=[1,b]$ can be decided as
follows:
\begin{enumerate}
  \item If $a<b$, then $a\preceq Y$: to see this, choose $\eps<(b-a)/3$ so
      that $f$ is strictly positive on a compact set $[b-\eps,b-2\eps]$
      (note that such a set must exist as $f$ is continuous and the support
      ranges until $b$). We can lower-bound the $k$-th moment of $Y$ as
      \begin{align*}
        \int_1^b y^kf(y)dy &\geq \big(\inf_{[b-2\eps,b-\eps]} f\big)\cdot \int_{b-2\eps}^{b-\eps} y^kdy\\
        & = \frac
1{k+1}\left[(b-\eps)^{k+1}-(b-2\eps)^{k+1}\right].
      \end{align*}
      Note that the infimum is positive as $f$ is strictly positive on the
      compact set $[b-2\eps,b-\eps]$. The lower bound is essentially an
      exponential function to a base larger than $a$, since $b-2\eps>a$,
  and thus (ultimately) grows faster than $a^k$.
  \item If $a>b$, then $Y\preceq a$, since $Y$ -- in any possible
      realization -- leads to strictly less damage than $a$. The formal
      argument is now based on an upper bound to the moments, which can be
      derived as follows:
      \begin{align*}
      \int_1^b y^k f(y)dy & \leq (\sup_{[1,b]} f)\cdot \int_1^b y^kdy = (\sup_{[1,b]} f)\frac 1{k+1}b^{k+1}.
      \end{align*}
      It is easy to see that for $k\to\infty$, this function grows slower
      than $a^k$ as $a>b$, which leads to the claimed $\preceq$-relation.
  \item If $a=b$, then we apply the mean-value theorem to the integral
      occurring in $\E{Y^k}=\int_1^a y^kf(y)dy$ to obtain an $\xi\in[1,a]$
      for which
      \[
        \E{Y^k}=\xi^k\underbrace{\int_1^af(y)dy}_{=1}=\xi^k\leq a^k
      \]
      for all $k$. Hence, $Y\preceq a$ in that case. An intuitive
      explanation stems from the fact that $Y$ may assign positive
      likelihood to events with less damage as $a$, whereas a deterministic
      outcome is always larger or equal to anything that $Y$ can deliver.
\end{enumerate}

\subsection{On the Interpretation of $\preceq$ and Inference}\label{sec:statistical-assurance}
The practical meaning of the $\preceq$-preference is more involved than just
a matter of comparing the first few moments. Indeed, unlike for IT risk
preferences based on \eqref{eqn:it-risk-formula}, the first moment can be
left unconstrained while $\preceq$ may still hold in either direction.

For general inference, the comparison of two distributions provides a
necessary basis (i.e., to define optimality, etc.). For example, (Bayesian)
decision theory or game theoretic models can be defined upon $\preceq$, via a
much deeper exploration of the embedding of $\F$ into the hyperreals (by
mapping a distribution to its moment sequence), such as the induced topology
and calculus based on it. In any case, however, we note that the previous
results may help in handling practical matters of $\preceq$ inside a more
sophisticated statistical decision or general inference process. For
practical decisions, some information can be obtained from the value $x_0$
that Theorem \ref{thm:tail-bounds} speaks about. This helps assessing the
meaning of the order, although the practical consequences implied by
$\leq_{st}$ or $\preceq$ are somewhat similar. The main difference is
\eqref{eqn:usual-stochastic-order} holding only for values $\geq x_0$ in case
of $\preceq$. The threshold can hence be found by numerically searching for
the largest (``right-most'') intersection point of the respective survival
functions; that is, for two distributions $F_1\preceq F_2$, a valid $x_0$ in
Theorem \ref{thm:tail-bounds} is any value for which $1-F_1(x)\leq 1-F_2(x)$
for all $x\geq x_0$. An approximation of $x_0$, e.g., computed by a bisective
search in common support of both distributions, then more accurately
describes the ``statistically best'' among the available actions, since
losses $>x_0$ are more likely for all other options. A practical decision, or
more general inference based on $\preceq$, should therefore be made upon
computing $x_0$ as an explicit auxiliary information, in order to assign a
quantitative meaning to ``extreme events'' in the interpretation underneath
Theorem \ref{thm:tail-bounds}. Further issues of practical decision making in
the context of IT risk management are discussed along the first empirical
example found in section \ref{sec:empirical-examples}.

Section \ref{sec:practical-examples} will not discuss (statistical) inference
since the details are beyond the scope of this work (we leave this to follow
up work). Instead, the following section will be dedicated to numerical
illustrations of $\preceq$ only, without assigning any decisional meaning to
the $\preceq$-preferred distributions. For each example, we will also give an
approximation (not the optimal) value of $x_0$.

\section{Numerical Examples}\label{sec:practical-examples}
Let us now apply the proposed framework to the problem of comparing effects
that are empirically measurable, when the precise action/response dynamics is
unknown. We start by looking at some concrete parametric models of extreme
value distributions first, to exemplify cases of numerical comparisons of
distributions with unbounded tails in Section
\ref{sec:examples-with-unbounded-tails}.

In Section \ref{sec:empirical-data}, we will describe a step-by-step
evaluation of our $\preceq$-ordering on empirical distributions. The sources
and context of the underlying empirical data sets are described in section
\ref{sec:empirical-examples}. From the data, we will compile non-parametric
distribution models, which are either normalized histograms or kernel density
estimators. On these, we will show how to decide the $\preceq$-relation using
the results from section \ref{sec:practicalities}.

\subsection{Comparing Parametric
Models}\label{sec:examples-with-unbounded-tails} We skip the messy algebra
tied to the verification of the criteria in Section
\ref{sec:unbounded-tails}, and instead compute the moments numerically to
illustrate the growth/divergence of moment sequences as implied by Lemma
\ref{lem:ordering-invariance}.

\begin{exa}[different mean, same variance]
Consider two Gumbel-distributions $X\sim F_1=Gumbel(31.0063, 1.74346)$ and
$Y\sim F_2=Gumbel(32.0063, 1.74346)$, where a density for $Gumbel(a,b)$ is
given by
\[
f(x|a,b) = \frac 1 b e^{\frac{x-a}{b}-e^{\frac{x-a}{b}}},
\]
where $a\in\R$ and $b>0$ are the location and scale parameter.

\begin{figure}[h!]
          \centering
            \includegraphics[scale=0.9]{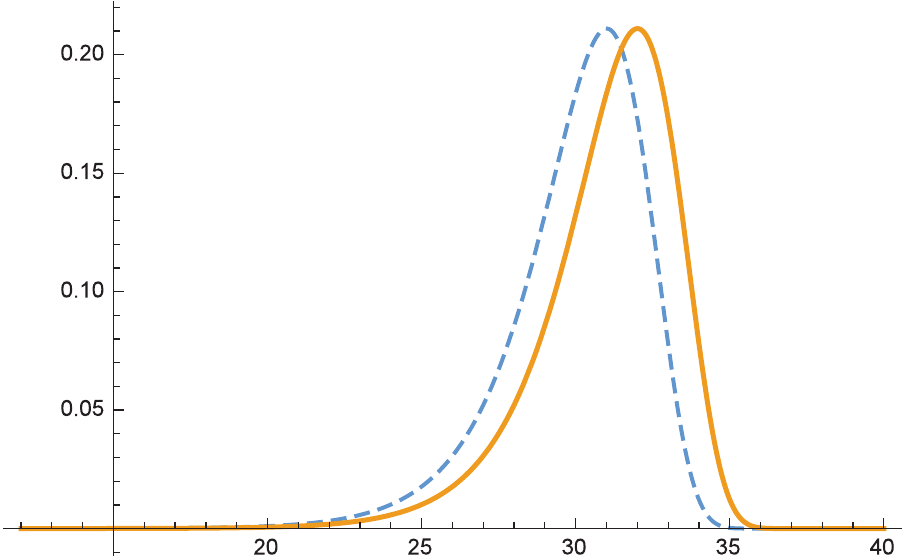}
          \caption{Comparing distributions with different means}\label{fig:different-mean-comparison}
        \end{figure}

Computations reveal that under the given parameters, the means are $\E{X}=30,
\E{Y}=31$ and $\Var{X}=\Var{Y}=5$. Figure \ref{fig:different-mean-comparison}
plots the respective densities of $F_1$ (dashed) and $F_2$ (solid line). The
respective moment sequences evaluate to
\begin{align*}
    \E{X^k}&=(30, 905, 27437.3, 835606, 2.55545\times 10^7,\ldots),\\
    \E{Y^k}&=(31, 966, 30243.3, 950906, 3.00162\times 10^7,\ldots),
\end{align*}
thus illustrating~ that $F_1\preceq F_2$. This is consistent with the
intuition that the preferred distribution gives \emph{less expected damage}.
The concrete region about which Theorem \ref{thm:tail-bounds} speaks is at
least for damages $>x_0=25$ (cf. Theorem \ref{thm:tail-bounds}).

\end{exa}

\begin{exa}[same mean, different variance]
Let us now consider two Gumbel-distributions $X\sim F_1=Gumbel(6.27294,
2.20532)$ and $Y\sim F_2=Gumbel(6.19073, 2.06288)$, for which $\E{X}=\E{Y}=5$
but $\Var{X}=8>\Var{Y}=7$.

\begin{figure}[h!]
          \centering
            \includegraphics[scale=0.9]{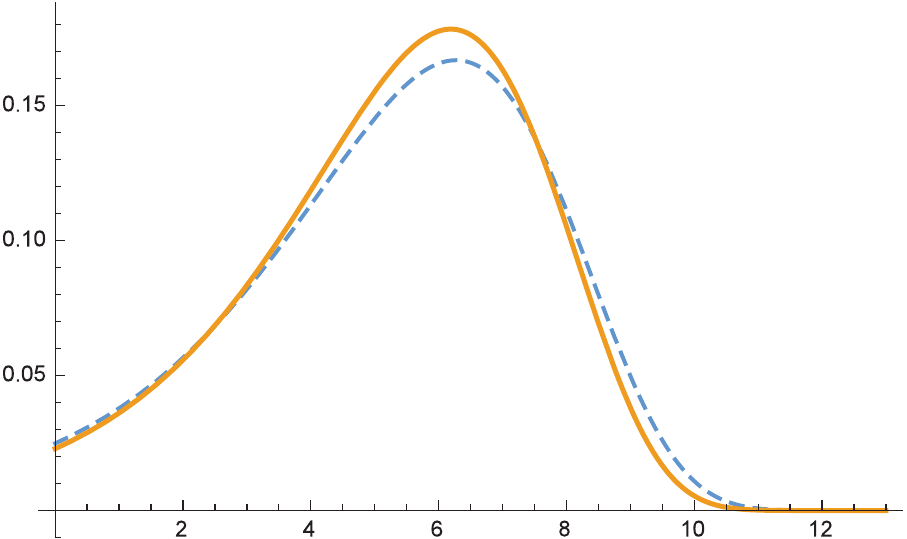}
          \caption{Comparing distributions with equal means but different variance}\label{fig:same-mean-comparison}
        \end{figure}

Figure \ref{fig:same-mean-comparison} plots the respective densities of $F_1$
(dashed) and $F_2$ (solid line). The respective moment sequences evaluate to
\begin{align*}
    \E{X^k}&=(5, 33, 219.215, 1654.9, 11957.8,\ldots),\\
    \E{Y^k}&=(5, 32, 208.895, 1517.51, 10806.8,\ldots),
\end{align*}
thus illustrating~ that $F_2\preceq F_1$. This is consistent with the
intuition that among two actions leading to the same expected loss, the
preferred one would be one for which the variation around the mean is
smaller; thus the loss prediction is ``more stable''.  The range on which
damages under $F_2$ are less likely than under $F_1$ begins at $x>x_0\approx
5.5$ (cf. Theorem \ref{thm:tail-bounds}).
\end{exa}

\begin{exa}[different distributions, same mean and
variance]\label{exa:equal-first-two-moments} Let us now consider a situation
in which the expected loss (first moment) and variation around the mean
(second moment) are equal, but the distributions are different in terms of
their shape. Specifically, let $X\sim F_1=Gamma(260.345, 0.0373929)$ and
$Y\sim Weibull(20,10)$, with densities as follows:

\[
f_{\text{Gamma}}(x|a,b)=
\left \{
\begin{array}{cc}
 \frac{b^{-a} x^{a-1} e^{-\frac{x}{b}}}{\Gamma (a)}, & x>0; \\
 0, & \text{otherwise} \\
\end{array}
\right.
\]

\[
f_{\text{Weibull}}(x|a,b)   =\left\{
\begin{array}{cc}
 \frac{a e^{-\left(\frac{x}{b}\right)^a} \left(\frac{x}{b}\right)^{a-1}}{b}, & x>0; \\
 0, & \text{otherwise} \\
\end{array}
 \right.
\]

\begin{figure}[h!]
  \centering
    \includegraphics[scale=0.9]{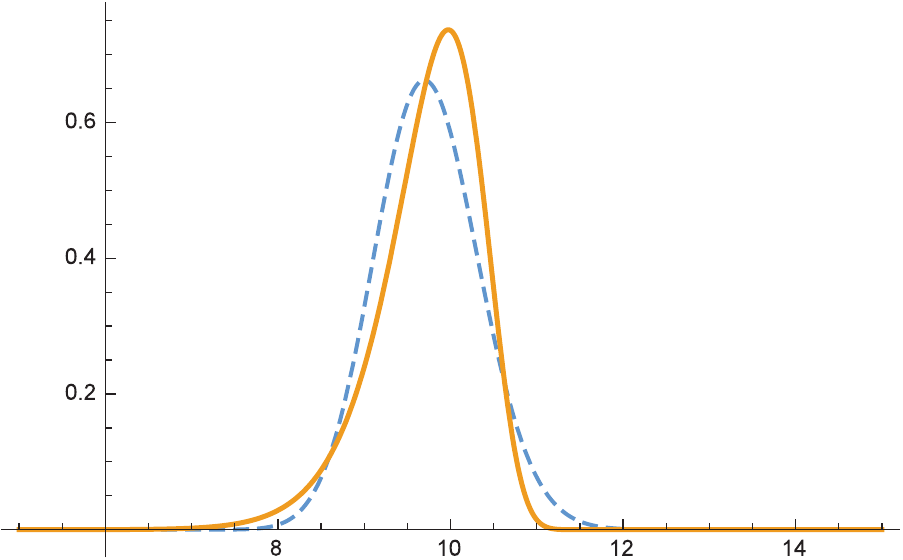}
  \caption{Comparing distributions with matching first two moments but different shapes}\label{fig:different-shape-comparison}
\end{figure}

 Figure \ref{fig:different-shape-comparison} plots the
respective densities of $F_1$ (dashed) and $F_2$ (solid line). The respective
moment sequences evaluate to
\begin{align*}
    \E{X^k}&=(9.73504, 95.1351, 933.259, 9190.01, 90839.7, \ldots),\\
    \E{Y^k}&=(9.73504, 95.1351, 933.041, 9181.69, 90640.2, \ldots),
\end{align*}
thus illustrating~ that $F_2\preceq F_1$. In this case, going with the
distribution that visually ``leans more towards lower damages'' would be
flawed, since $F_1$ nonetheless assigns larger likelihood to larger damages.
The moment sequence, on the contrary, unambiguously points out $F_2$ as the
preferred distribution (the third moment tips the scale here; cf.
\cite{Eichner&Wagener2011,Chiu2010,Wenner2002}). The statistical assurance
entailed by Theorem \ref{thm:tail-bounds} about an interval in which high
damage incidents are
less likely (at least) includes losses $>x_0\approx 10.3$. 

\end{exa}

\subsection{Empirical Test Data and Methodology}\label{sec:empirical-examples}
To demonstrate how the practical matters of comparing distributions work, we
will use three sets of empirical data, based on qualitative data from risk
estimation, and based on simulating a malware outbreak using
percolation.\todo{the context descriptions of the test data sets \#1 and \#2
were rewritten towards being more compact and tight.}

\paragraph{Test Data Set \#1 -- IT Risk Assessments:}
The common quantitative understanding of risk by the formula
\eqref{eqn:it-risk-formula} is easily recognized as the \emph{expectation}
(i.e., first moment) of a loss distribution. Although being standard in
quantitative IT risk management, its use is discouraged by the German Federal
Office of Information Security (BSI) \cite{Muench2012} for several reasons
besides the shortcomings that we discussed here (for example, statistical
data may be unavailable at the desired precision and an exact formula like
\eqref{eqn:it-risk-formula} may create the illusion of accuracy where there
is none \cite{Muench2012}).


Best practices in risk management (ranging up to norms like the ISO27005
\cite{ISO2011}, the ISO31000 \cite{ISO2009} or the OCTAVE Allegro framework
\cite{Caralli2007}) usually recommend the use of \emph{qualitative} risk
scales. That is, the expert is only asked to utter an opinion about the risk
being ``low/medium/high'' or perhaps using a slightly more fine-grained but
in any case ordinal scale. In a slight abuse of formalism, these categories
are then still carried into an evaluation of \eqref{eqn:it-risk-formula} (cf.
\cite{Goodpasture2004}) towards finding the decision with the ``least'' risk
in qualitative terms as \eqref{eqn:it-risk-formula} gives.

Categorical risk assessments are heavily used in the IT domain due to their
good systematization and tool support. Our first test data set is thus a risk
assessment made in terms of the Common Vulnerability Scoring System (CVSS)
\cite{Mell2007}. The CVSS ranks risks on a scale from 0 to 10, as a decimal
rounded up to one place behind the comma. Usually, these CVSS values come
from domain experts, so there is an intrinsic ambiguity in the opinions on
grounds of which a decision shall be made. Table \ref{tbl:cvss-example},
taken from \cite{Beck2016} (by kind permission of the author A.
Beck), shows an example of such expert data for two security system
installments being assessed by experts in the left and right part of the
table (separated by the double vertical line). The $\preceq$-ordering shall
now help to choose the better of the two options, based on the ambiguous and
even inconsistent domain expert inputs. For simplicity of the example, we did
not work with the fine-grained CVSS scores, but coarsened them into three
categories, i.e., intervals of scores $low=[0,3)$ (L),
$medium=[4,8)$ (M) and $high=[8,10]$ (H). We remark that the categorial assessment was added in this work, and is
	not from the source literature.

%
%

\begin{table}[t!]
  \centering
  \begin{tabular}{|p{1.2cm}|l|l|l|l|l|l||l|l|l|l|l|l|l|l|l|}
  \hline
  Expert (anony\-mized) & \rot{CN-863} & \rot{ER-881} & \rot{\"OL-968} & \rot{BA-576} & \rot{RC-813} & \rot{RR-745} & \rot{EN-720} & \rot{EF-375} & \rot{UE-941} & \rot{RI-740} & \rot{UM-330} & \rot{TR-790} & \rot{EE-677} & \rot{ER-640} & \rot{EE-489} \\\hline
Scenario & 	1 & 1 & 1	 & 1	 & 1	 & 1	 & 2	 & 2	 & 2 & 	2	 & 2 & 	2 & 	2 & 	2 & 	2 \\\hline
CVSS & 	10 & 6.4 & 9	 & 7.9	 & 7.1 & 	9	 & 10	 & 7.9	 & 8.2	 & 7.4	 & 10 & 	8.5	 & 9	 & 9 & 	8.7 \\\hline
Risk & H & M & H & M & M & H & H & M & H & M & H & H & H & H & H\\
  \hline
\end{tabular}
  \caption{Example CVSS Risk Assessment \cite{Beck2016}}\label{tbl:cvss-example}
\end{table}

%

\paragraph{Test Data Set \#2 -- Malware Outbreaks:}
Computer malware infections are continuously reported in the news, with an
early and prominent example having been the Stuxnet worm in 2008
\cite{Karn11}, which infected the Iranian uranium enrichment facilities. Ever
since, the control and supervision of cyber-physical systems has gained much
importance in risk management, since attacks on the computer infrastructure
may have wide effects ranging up to critical supply infrastructures such as
water supply, power supply, and many others (e.g., oil, gas or food supply networks, etc.).

The general stealthiness of such infections makes an exact assessment of risk
difficult. A good approach to estimate that risk is to apply outbreak simulation models, such
as, for example, using percolation theory \cite{Newman2002,Koenig2016}. These
simulations provide us with possible infection scenarios, in which the number
of infected nodes (after a fixed period of time), can be averaged into a
probability distribution describing the outcome of an infection. Repeating
the simulation with different system configurations yields various outcome
distributions. An example for a network with 20 nodes and 1000 repetitions per
simulation is displayed in Table \ref{tbl:malware-data}. The
$\preceq$-relation shall then help deciding which configuration is better in
minimizing the risk of a large outbreak.

\begin{table}
  \centering

    \begin{tabular}{|r|c|c|}\hline
size of the outbreak & config. 1 & config. 2\\\hline
1	&0	  &0      \\
2	&1	  &0      \\
3	&0	  &0      \\
4	&3	  &0      \\
5	&5	  &3      \\
6	&1	  &4      \\
7	&4	  &6      \\
8	&5	  &5      \\
9	&9	  &8      \\
10&	6	 & 25    \\
11&	13	&33     \\
12&	22	&39     \\
13&	29	&85     \\
14&	44	&131    \\
15&	86	&160    \\
16&	135&	164   \\
17&	182&	150   \\
18&	245&	113   \\
19&	173&	64    \\
20&	37	&10     \\

\hline
    \end{tabular}

  \caption{Simulated malware infection: number of occurrences of an outbreak of size $n$ under two configurations, after a fixed time period}\label{tbl:malware-data}
\end{table}


\paragraph{Test Data Set \#3 -- Nile Water Level:}
As a third data set, we use one that ships with the statistical software
suite \texttt{R}. Concretely, we will look at the dataset \texttt{Nile} that
consists of the measurements of the annual flow of the Nile river between
1871 and 1970. For comparisons, we will divide the data into two groups of 50
observations each (corresponding to years). The decision problem associated
with it is the question of which period was more severe in terms of water
level. Extending the decision problem to more than two periods would then
mean searching for a \emph{trend} within the data. Unlike a numerical trend,
such as a sliding mean, we would here have a ``sliding empirical
distribution'' to determine the trend in terms of randomness.

\subsection{Comparing Empirical Distributions}\label{sec:empirical-data}
With the three data sets as described, let us now look into how decisions
based on the empirical data can be made.

\paragraph{Categorical Data -- Comparing Normalized Histograms:}

Compiling an empirical distribution from the example CVSS data in Table
\ref{tbl:cvss-example} gives the histograms shown in Figure
\ref{fig:barplots}. Clearly, scenario 1 is preferable here, as it is less
frequently rated with high damage than scenario 2. On the contrary, the
decision is much less informed than in the case where the full numeric data
would have been used. We will thus revisit this example later again.


\begin{figure}
  \centering
  \includegraphics[scale=1]{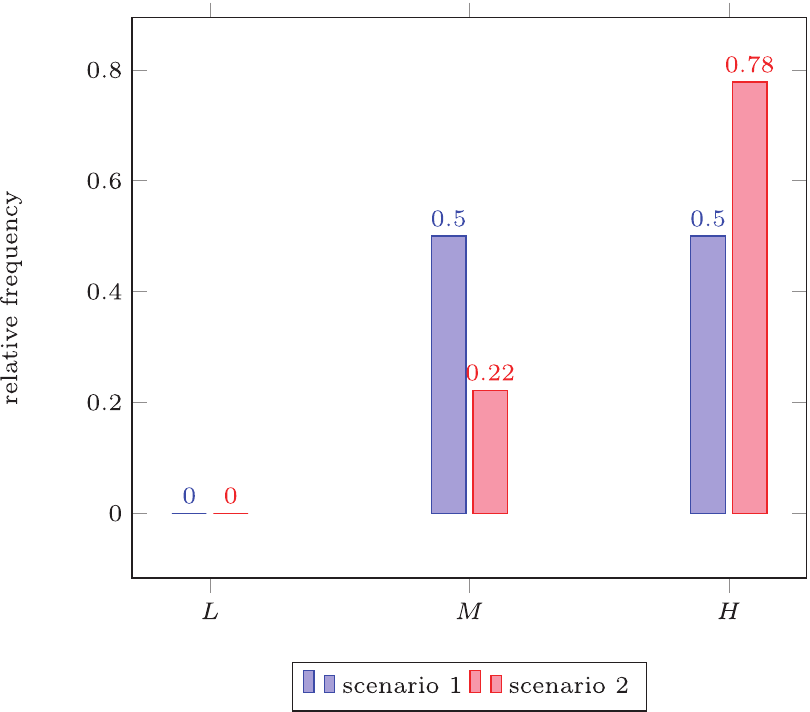}
  \caption{Comparing Empirical Distributions -- Test Data \#1 in Nominal Scale\\
  Outcome: ``scenario 1'' $\prec$ ``scenario 2'' (i.e., scenario 1 has lower
  security risk, based on the coarsened data), for medium and high categories ($x_0=~$'$M$').}\label{fig:barplots}
\end{figure}

For the simulated malware infection data in Table \ref{tbl:malware-data}, the
empirical distribution of the number of affected nodes as shown in Figure
\ref{fig:malware-counts-frequ} is obtained by normalization of the
corresponding histograms. In this case, configuration 2 is preferable as the
maximal damage of 20 node has occurred less often than in configuration 1.



\begin{figure}
  \centering
  \includegraphics[scale=1]{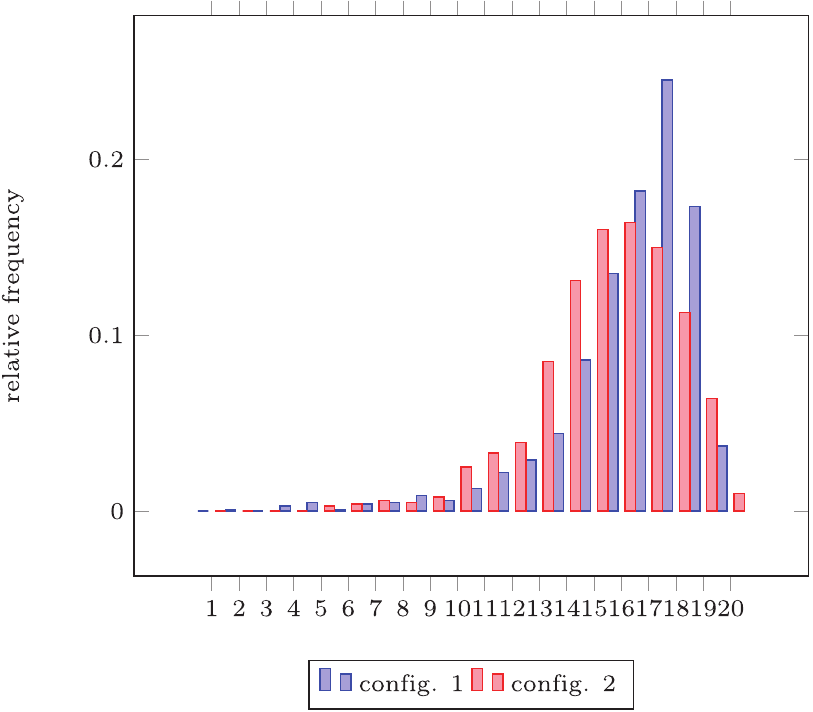}
\caption{Test Data \#2 -- (Simulated) Empirical Distribution of Malware
Outbreak Sizes under two Configurations\\Outcome: ``config. 2'' $\prec$ ``config. 1''
(i.e., the second configuration is more secure w.r.t. extreme outcomes, i.e.,
infections of $>x_0=9$ nodes)}\label{fig:malware-counts-frequ}
\end{figure}

\paragraph{Continuous Data -- Comparing Kernel Density Estimates:}
If the data itself is known to be continuous, then a nonparametric
distribution estimate can be used to approximate the unknown distribution.

In the following, let us write $\hat f$ to mean a general kernel density
estimate based on the data (observations) $x_1, \ldots, x_n$, of the form
\begin{equation}\label{eqn:kde}
\hat f_n(x) = \frac 1{n\cdot h}\sum_{i=1}^n K\left(\frac{x-x_i}h\right),
\end{equation}
where $K(x)$ is the chosen kernel function, and $h>0$ is a bandwidth
parameter, whose choice is up to any (of many existing) heuristics (see
\cite{Liu2009,Silverman1998} among others). Computing a kernel density
estimate from data is most conveniently done by invoking the \texttt{density}
command within the \texttt{R} statistical computing software \cite{RDCT2016}.

This comparison of two kernel density estimates (KDE) is illustrated in
Figure \ref{fig:KDE-comparison}. For that purpose, we divided the test
dataset \#3 (\texttt{data(Nile)} in \texttt{R}) into observations covering
the years 1871-1920 and 1921-1970. For both sets the density is estimated
with a Gaussian kernel and the default bandwidth choice \texttt{nrd0}
(Silverman's rule \cite{Silverman1998}) yielding a KDE $\hat{f}_1$ with
bandwidth $h_1=79.32$ for the years 1871-1920 and a KDE $\hat{f}_2$ with
bandwidth $h_2=45.28$ for the years 1921-1970. Further we have the maximal
observed values $x_{n_1}=1370$ and $y_{n_2}=1170$, and see that
\[
x_{n_1}+h_1=1449.32 > y_{n_2}+h_2 = 1215.28.
\]
Therefore (and also by visual inspection of Figure \ref{fig:KDE-comparison}),
the density for the period from 1871 until 1920 has had higher likelihoods
for a high water level, which became less in the period from 1921-1970, thus
indicating a ``down-trend'' by $\hat{f}_2 \prec \hat{f}_1$.

\begin{figure}
  \centering
  \includegraphics[scale=1]{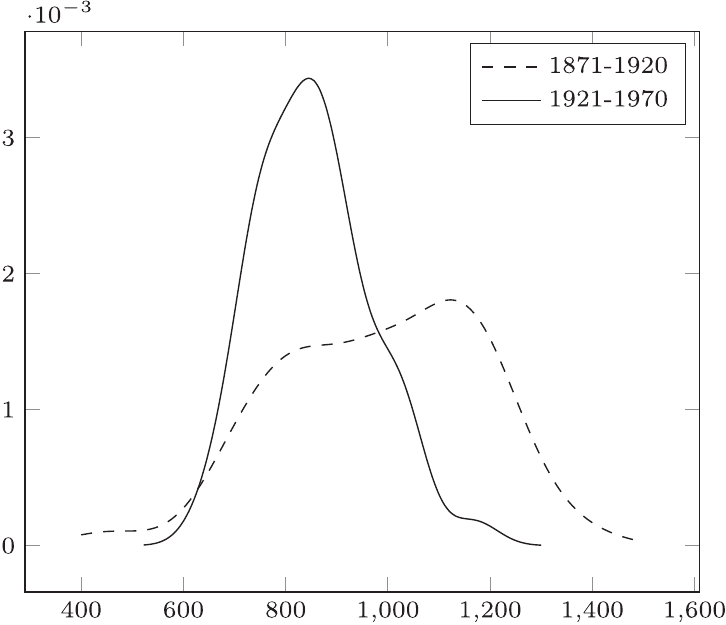}
  \caption{Comparison of two kernel density estimates -- Test Data \#3
  (Nile Water Level)\\Outcome: ``1921-1970'' $\prec$ ``1871-1920''
  (indicating that floodings at a water level $>x_0\approx 220$ (Theorem \ref{thm:tail-bounds}) have been more likely in the years before 1921).}\label{fig:KDE-comparison}
\end{figure}

\paragraph{Using Gaussian Kernels:}
Commonly, the kernel density approximation is constructed using Gaussian
kernels per default, in which case $K$ takes the form $K(x)=\frac
1{\sqrt{2\pi}}\exp(-\frac 1 2 x^2)$. Definition
\ref{asm:finity-of-repair-costs} is clearly not met, but there is also no
immediate need to resort to the extended version of $\preceq$ as given by
Definition \ref{def:extended-preference}. Indeed, if we simply truncate the
KDE at any point $a>1$ into the distribution $\hat f$ and remember that $K\in
C^\infty$, the truncated kernel density estimate is again a
$C^\infty$-density, as required by Lemma \ref{lem:derivative-comparisons}.

Returning to the CVSS example data in Table \ref{tbl:cvss-example}, we
constructed Gaussian kernel density estimates $f_1, f_2$, with bandwidths
$h_1\approx 0.798$ and $h_2\approx 0.346$ (using the the default Silverman's
rule in \texttt{R}); plots of which are given in Figure \ref{fig:kde-plots}.
Using the criterion of Lemma \ref{lem:derivative-comparisons} in connection
with the lexicographic ordering, we end up finding that $f_2\preceq f_1$, in
contrast to our previous finding. This is, however, only an inconsistency at
first glance, and nevertheless intuitively meaningful if we consider the
context of the decision and the effect of the nonparametric estimation more
closely:
\begin{itemize}
  \item Since scenario 2 is based on more data than scenario 1, the
      bandwidth $h_2$ is less than $h_1$. This has the effect of the
      distribution being ``more condensed'' around higher categories, as
      opposed to the distribution for scenario 1, whose tail is much
      thicker. Consequently, the decision is to prefer scenario 2 is
      implicitly based on the larger data set, and considers the higher
      uncertainty in the information about scenario 1.
  \item The Gaussian kernel has tails reaching out to $+\infty$, which also
      assigns positive mass to values outside the natural range of the
      input data ($1\ldots 10$ in case of CVSS). In many contexts,
      observations may not be exhaustive for the possible range (e.g.,
      monetary loss up to the theoretical maximum may -- hopefully -- not
      have occurred in a risk management process in the past). By
      construction, the KDE puts more mass on the tails the more data in
      this region is available. From a security perspective, this mass
      corresponds to \emph{zero-day exploit events}. Thus, such incidents
      are automatically accounted for by $\preceq$.
\end{itemize}

\begin{figure}
  \centering
  \includegraphics[scale=1]{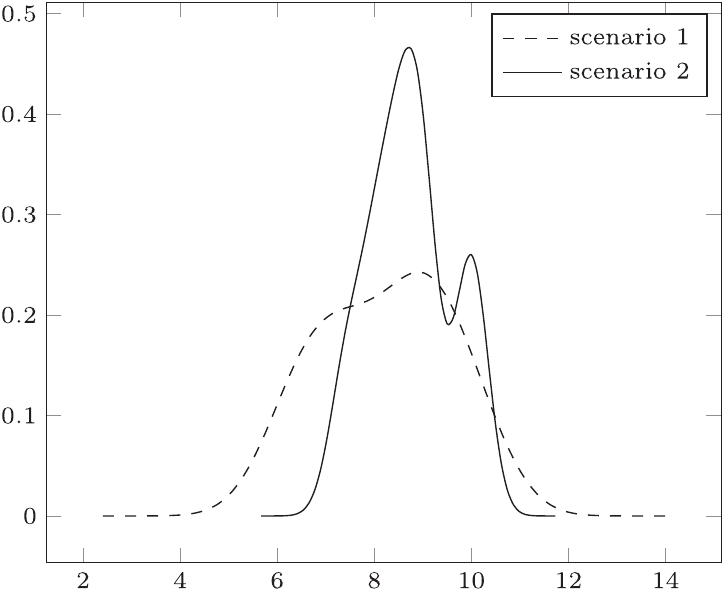}
\caption{Comparing Nonparametric Distribution Models -- Test Data \#1 (directly used)\\Outcome: ``scenario 2'' $\prec$ ``scenario 1'' (given a more refined view that in Figure \ref{fig:barplots}, scenario 1 has less security)}\label{fig:kde-plots}
\end{figure}

\section{Discussion}
\todo{the discussion has been rewritten towards being more tight and
	referring closer to the method laid out in the paper.}

Our proposed preference relation is designed for IT risk management. In this
context, decision makers often rely on scarce and purely subjective data
coming from different experts. Bayesian techniques that need large amounts of
data are therefore hard to apply (and somewhat ironically, a primary goal of
IT risk management is exactly minimizing the lot of incidents that could
deliver the data). Since the available information may not only be vague but
possibly also inconsistent, consensus finding by data aggregation is often
necessary. There exist various non-probabilistic methods to do this (such as
fuzzy logic, Dempster-Shafer theory, neural networks, etc.) and perform
extremely well in practice, but the interpretation of the underlying concepts
is intricate and the relation to values and business assets is not trivial.

To retain interpretability, data aggregation often means averaging (or taking
the median of) the available risk figures, in order to single out an optimal
action. This clearly comes at the cost of losing some information. Stochastic
orders elegantly tackle the above issues by letting the entire data go into a
probability distribution (and thus preserving all information), and defining
a meaningful ranking on the resulting objects. However, not all stochastic
orders are equally meaningful for the peculiarities of IT risk management.
For example, low damages are normally disregarded as being covered by the
system's natural resilience, i.e., no additional efforts are put on lowering
a risk that is considered as low already. The relevance of risks depends on
whether or not a certain acceptable damage threshold is exceeded. IT risk
managers typically care about significant (extreme) distortions and events
with high potential of damage but with only a limited lot of reported
evidence so far, such as zero day exploits or advanced persistent threats.

Consequently, a suitable ordering may reasonably ignore damages of low
magnitude, and focus on extreme outcomes, i.e., the tails of the respective
loss distributions. This is a major reason for our transition from the usual
stochastic order that takes into account the entire loss range (in fact all
$\R$, according to \eqref{eqn:usual-stochastic-order}) to one that explicitly
focuses on a left neighborhood of the loss maximum. In a converse approach to
the same problem, this could as well be used as a starting point to define an
order, but starting from moments instead and finishing with an ordering that
is about the heaviness of tails is an interesting lesson learned from our
proposed technique of using $^*\R$ to construct the ordering here. More
importantly, the rich structure of $^*\R$, being available without additional
labor, makes our ordering useable with optimization and game theory, so that
important matters of security economics can be covered as a by-product. This
non-standard technique of constructing an ordering is an independent
contribution of this work.\todo{The reviewer's comment on the possibility of
avoiding non-standard
	analysis is certainly true, and we acknowledge this remark here explicitly.
	Nevertheless, drawing fellow researcher's attention to this possibility
	appears important to us.}

Summarizing our point, decision making based on a stochastic ordering has the
appeal of a statistical fundament that is easy to communicate and, more
importantly, fits well into existing risk management standards (ISO 27000,
ISO 31000, etc.).

\paragraph{Outlook:}
The well defined arithmetic over $^*\R$, into which Theorem
\ref{thm:ordering-invariance} embeds the (risk) distribution models in $\F$,
lets us technically work with distributions like as if we were in a
topological field. This embedding offers an interesting unexplored (and
nontrivial) route of future research: though the operations on random
variables (say, addition or quotients) do not correspond to the same
operations in $^*\R$ (which is immediately evident from the definition), many
other operations and even functions of random variables can be studied in the
space $^*\R$ rather than on the set of distributions. So we can, for example,
do optimization theory over distributions but equipped with the full armory
of calculus known from the reals (that analogously holds in the space $^*\R$
by virtue of {\L}os' theorem or the transfer principle \cite{Robinson1966}).

Our ordering relation on the set of probability distributions can be extended
towards a theory of games on these spaces (this extension is based on the
topology that the order induces, upon which Nash's result on the existence of
equilibria can be re-established on our space of probability distributions).
First steps into applying the framework to competitive decision-situations
have been taken in \cite{Rass2015c}, and will be further detailed in follow
up research articles.







\section{Acknowledgments}
The authors wish to thank the anonymous reviewers for invaluable suggestions
and for bringing up a variety of interesting aspects to look at here and
future research. Their input greatly improved the readability and quality of
the text. This work was supported by the European Commission's Project No.
608090, HyRiM (Hybrid Risk Management for Utility Networks) under the 7th
Framework Programme (FP7-SEC-2013-1). The project is online found at
\url{https://hyrim.net}.


\section{References}

\begin{appendix}

\section{Proofs} The proofs here
first appeared in \cite{Rass2015a}, and are repeated for the sake of
completeness and convenience of the reader.
\subsection{Proof of Lemma
\ref{lem:ordering-invariance}}\label{apx:proof-of-ordering-invariance} We
first discuss the continuous case, which illustrates the basic idea that can
be applied alike to categorical and discrete distributions.

Let $f_1, f_2$ denote the densities of the distributions $F_1,F_2$. Fix the
smallest $b^*>1$ so that $\Omega:=[1,b^*]$ covers both the supports of $F_1$
and $F_2$. Consider the difference of the $k$-th moments, given by
\begin{align}
    \Delta(k) := \E{R_1^k}-\E{R_2^k} &= \int_{\Omega}x^k f_1(x)dx - \int_{\Omega}x^k f_2(x)dx\nonumber\\
    &=\int_{\Omega}x^k(f_1-f_2)(x)dx\label{eqn:difference-of-moments}.
\end{align}
Towards a lower bound to \eqref{eqn:difference-of-moments}, we distinguish
two cases:
\begin{enumerate}
  \item If $f_1(x)>f_2(x)$ for all $x\in\Omega$, then $(f_1-f_2)(x)>0$ and
      because $f_1,f_2$ are continuous, their difference attains a minimum
      $\lambda_2>0$ on the compact set $\Omega$. So, we can lower-bound
      \eqref{eqn:difference-of-moments} as
      $\Delta(k)\geq\lambda_2\int_{\Omega}x^kdx\To+\infty$, as
      $k\To\infty$.
  \item Otherwise, we look at the right end of the interval $\Omega$, and
      define
        \[
            a^* :=\inf\set{x\geq 1: f_1(x)>f_2(x)}.
        \]
     Without loss of generality, we may assume $a^*<b^*$. To see this, note
      that if $f_1(b^*)\neq f_2(b^*)$, then the continuity of $f_1-f_2$
      implies $f_1(x)\neq f_2(x)$ within a range $(b^*-\eps,b^*]$ for some
      $\eps>0$, and $a^*$ is the supremum of all these $\eps$. Otherwise,
      if $f_1(x)=f_2(x)$ on an entire interval $[b^*-\eps,b^*]$ for some
      $\eps>0$, then $f_1\not>f_2$ on $\Omega$ (the opposite of the
      previous case) implies the existence of some $\xi<b^*$ so that
      $f_1(x)<f_2(x)$, and $a^*$ is the supremum of all these $\xi$  (see
      Figure \ref{fig:density-lower-bounds} for an illustration). In case
      that $\xi=0$, we would have $f_1\geq f_2$ on $\Omega$, which is
      either trivial (as $\Delta(k)=0$ for all $k$ if $f_1=f_2$) or
      otherwise covered by the previous case.

      In either situation, we can fix a compact interval $[a,b]\subset
      (a^*,b^*)\subset[1,b^*]=\Omega$ and two constants
      $\lambda_1,\lambda_2>0$ (which exist because $f_1,f_2$ are bounded as
      being continuous on the compact set $\Omega$), so that the function
      \[
        \ell(k,x) := \left\{
                     \begin{array}{rl}
                       -\lambda_1x^k, & \hbox{if }1\leq x<a; \\
                       \lambda_2x^k, & \hbox{if }a\leq x\leq b.
                     \end{array}
                   \right.
      \]
        lower-bounds the difference of densities in
        \eqref{eqn:difference-of-moments} (see Figure
        \ref{fig:density-lower-bounds}), and

    \begin{figure}
          \centering
            \includegraphics[scale=0.9]{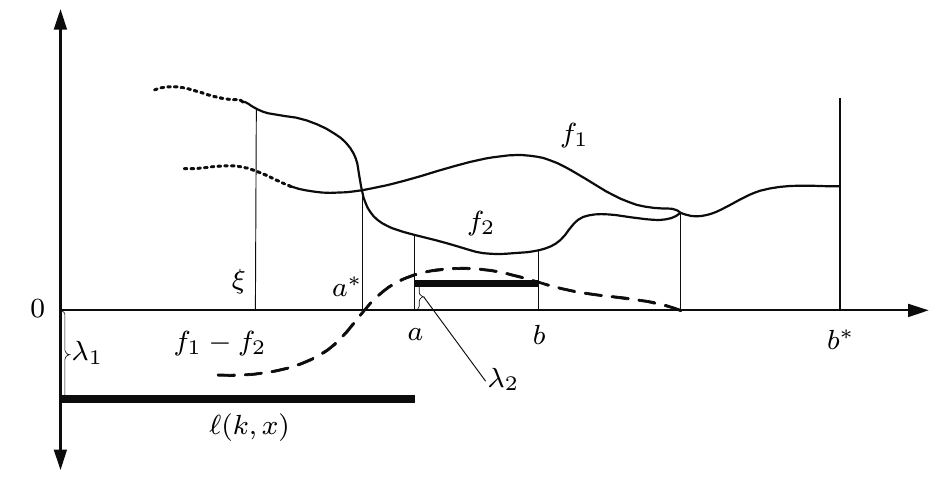}
          \caption{Lower-bounding the difference of densities}\label{fig:density-lower-bounds}
        \end{figure}

      \begin{align*}
        \Delta(k)=\int_1^{b^*}x^k(f_1-f_2)(x)dx &\geq \int_1^b\ell(x,k)dx \\
        &=-\lambda_1\int_1^ax^kdx + \lambda_2\int_a^bx^kdx\\
        &=-\frac{a^{k+1}}{k+1}(\lambda_1+\lambda_2)+\lambda_2\frac{b^{k+1}}{k+1}\To+\infty,
      \end{align*}
        as $k\To\infty$ due to $a<b$ and because $\lambda_1,\lambda_2$ are
        constants that depend only on $f_1,f_2$.

    In both cases, we conclude that, unless $f_1=f_2$, $\Delta(k)>0$ for
    sufficiently large $k\geq K$ where $K$ is finite. This establishes the
    lemma for continuous distributions.
\end{enumerate}

In the discrete or categorical case, the argument remains the same, only
adapted to looking at the finite set of values on which $f_1\geq f_2$. The
largest value less than $a$ above which equality holds until the end of the
support then determines the growth of the difference sequence in the same way
as was argued in Section
\ref{sec:preference-between-categorical-distr}.\hfill$\square$

\subsection{Proof of Theorem \ref{thm:tail-bounds}}\label{apx:tail-bound}
Let $f_1, f_2$ be the density functions of $F_1, F_2$. Call
$\Omega=\supp(F_1)\cup\supp(F_2)=[0,a]$ the common support of both densities,
and take $\xi=\inf\set{x\in\Omega: f_1(x)=f_2(x)=0}$. Suppose there were an
$\eps>0$ so that $f_1>f_2$ on every interval $[\xi-\delta,\xi]$ whenever
$\delta<\eps$, i.e., $f_1$ would be larger than $f_2$ until both densities
vanish (notice that $f_1=f_2=0$ on the right of $\xi$). Then the proof of
lemma \ref{lem:ordering-invariance} delivers the argument by which we would
find a $K\in\N$ so that $\E{X_1^k}>\E{X_2^k}$ for every $k\geq K$, which
would contradict $F_1\preceq F_2$. Therefore, there must be a neighborhood
$[\xi-\delta,\xi]$ on which $f_1(x)\leq f_2(x)$ for all $x\in
[\xi-\delta,\xi]$. The claim follows immediately by setting $x_0=\xi-\delta$,
since taking $x\geq x_0$, we end up with $\int_x^\xi f_1(t)dt\leq
\int_{x}^\xi f_2(t)dt$, and for $i=1,2$ we have $\int_x^\xi f_i(t)dt =
\int_x^a f_i(t)dt = \Prob{X_i>x}$.\hfill$\square$

\subsection{Proof of Lemma
\ref{lem:approximation-comparison}}\label{apx:criterion} Throughout the
proof, let $i\in\set{1,2}$. The truncated distribution density that
approximates $f_i$ is $f_i(x)/(F_i(a_n)-F_i(0))$, where $[0,a_n]$ is the
common support of $n$-th approximation to $f_1, f_2$. By construction,
$a_{n,i}\To\infty$ as $n\To\infty$, and therefore $F_i(a_n)-F_i(0)\To 1$ for
$i=1,2$. Consequently,
\[
    Q_n = \frac{F_1(a_n)-F_1(0)}{F_2(a_n)-F_2(0)}\To 1,\quad\text{ as }n\To\infty,
\]
and there is an index $N$ such that $Q_n > c$ for all $n\geq N$. In turn,
\[
    f_2(x)\cdot Q_n > f_2(x)\cdot c> f_1(x),
\]
and by rearranging terms,
\begin{equation}\label{eqn:truncated-comparison}
    \frac{f_1(x)}{F_1(a_n)-F_1(0)} < \frac{f_2(x)}{F_2(a_n)-F_2(0)},
\end{equation}
for all $x\geq x_0$ and all $n\geq N$. The last inequality
\eqref{eqn:truncated-comparison} lets us compare the two approximations
easily by the same arguments as have been used in the proof of lemma
\ref{lem:ordering-invariance}, and the claim follows.


\end{appendix}

\end{document}